\documentclass[a4paper,authoryear,12pt]{article}

\RequirePackage[OT1]{fontenc}
\usepackage{amsmath,supertabular,tabularx,amssymb,amsfonts,amsthm}
\usepackage[pdftex]{color,graphicx}
\usepackage[pdftex,bookmarks,colorlinks]{hyperref}
\hypersetup{pdfauthor={somakd},citecolor=blue}
\usepackage{rotating}
\usepackage{natbib}
\usepackage[mathlines]{lineno}
\linenumbers
\allowdisplaybreaks[1]

\renewcommand{\a}{\ensuremath{\alpha}}
\renewcommand{\b}{\ensuremath{\beta}}
\newcommand{\g}{\ensuremath{\gamma}}

\newcommand{\e}{\ensuremath{\epsilon}}

\renewcommand{\l}{\lambda}
\newcommand{\dint}{\int\displaylimits}

\newcommand{\be}{\pmb\e}
\newcommand{\s}{\sigma}
\newcommand{\B}{\ensuremath{\mathcal{B}}}
\theoremstyle{plain}
\newtheorem{theorem}{Theorem}
\newtheorem{lemma}{Lemma}
\newtheorem{corr}{Corollary}

\newcommand{\topline}{\hrule height 1pt width \textwidth \vspace*{2pt}}
\newcommand{\botline}{\vspace*{2pt}\hrule height 1pt width \textwidth \vspace*{4pt}}
\newtheorem{algo}{Algorithm} 

\numberwithin{equation}{section}
\numberwithin{algo}{section}
\numberwithin{table}{section}
\numberwithin{figure}{section}

\newcommand{\statesp}{\ensuremath{\mathcal X}}
\newcommand{\Y}{\ensuremath{\mathcal Y}}

\newcommand{\intersect}{\cap}
\newcommand{\imply}{\Longrightarrow}
\newcommand{\R}{\ensuremath{\mathbb R}}
\newcolumntype{Z}{>{\centering \arraybackslash}X}
\title{Multiplicative random walk Metropolis-Hastings on the real line}
\author{Somak Dutta \\ {\small University of Chicago, IL, USA.}} 

\date{}
\begin{document}
\maketitle
\pagestyle{myheadings}
\markboth{S. Dutta}{Random dive MH}


\begin{abstract}
In this article we propose multiplication based random walk Metropolis Hastings (MH) algorithm on the real line. We call it the random dive MH (RDMH) algorithm. This algorithm, even if simple to apply, was not studied earlier in Markov chain Monte Carlo literature. One should not confuse RDMH with RWMH. It is shown that they are different, conceptually, mathematically
and operationally.  The kernel associated with theRDMH algorithm is shown to have standard properties like irreducibility, aperiodicity and Harris recurrence under some mild assumptions. These ensure basic convergence (ergodicity) of the kernel. Further the kernel is shown to be geometric ergodic for a large class of target densities on \R. This class even contains realistic target densities for which random walk or Langevin MH are not geometrically ergodic. Three simulation studies are given to demonstrate the mixing property and superiority of RDMH to standard MH algorithms on real line. A share-price return data is also analyzed and the results are compared with those available in the literature.
\end{abstract}


\textbf{AMS classification number} \emph{Primary}: 65C05, 65C40. \emph{Secondary}: 60J10.


\section{Introduction}
Suppose $\pi$ is a density on a state-space $\statesp$ with respect to some dominating measure $\lambda$. 
Most often the state-space is a subset of the Euclidean space and $\lambda$ is the Lebesgue measure. 
In this article we will always assume $\lambda$ to be the Lebesgue measure. Statisticians' main aim is to study the characteristics of the density $\pi$. Sometimes (say, in Bayesian inference) $\pi$ may be a complicated (possibly unnormalized) density which is not analytically tractable. So, to study the characteristics of $\pi$, statisticians try to draw a sample from $\pi$. But then also there may not exist any effective simulation procedure to simulate from $\pi$. Thus the goal is shifted to draw an approximate sample from $\pi$. Markov chain Monte Carlo (MCMC) provides a method for doing this. The MCMC methods, quite famous for their effectiveness in drawing an approximate sample from a target density are widely used. One of the most famous MCMC algorithms is the Metropolis-Hastings (MH) algorithm \citep{metropolis1953,hastings1970}. Given a current state $x \in \statesp$ the MH algorithm proposes a new state $y$ from a \emph{proposal kernel density} $q(x\to y)$ on $\statesp$, and accepts it with the acceptance probability
\begin{equation}\label{eqn:MHacceptprob}
 \alpha(x \to y) = \min\left\{ \frac{\pi(y)q(y\to x)}{\pi(x)q(x\to y)},~1\right\}
\end{equation}
The corresponding MH kernel has stationary distribution $\pi(\cdot)$.

The random walk MH (RWMH) has the proposal kernel $q(x\to y) = q(|x-y|)$. We notice that generation a point from such a density is same as generating an $\e$ from $q(\e)$ and then setting $y = x + \e$. Other algorithms like Langevin MH (LMH), has proposal kernel $N(x + (\sigma^2/2)\nabla \log~\pi(x),\sigma^2)$

All of the aforementioned algorithms have certain disadvantages. For example, both the RWMH and LMH have slow mixing rates in certain cases. The acceptance rate for RWMH is high if $q(\cdot)$ is concentrated around zero (i.e. \emph{small step size}) but then a long chain is required to explore a substantial part of the state space. If a diffused proposal is used then the acceptance rate drops. For a multi--modal target (which is not known a priori in most cases) both the RWMH and the LMH chain may remain stuck at one or few of the modes and may  still pass convergence diagnostics. Obviously any inference based on such samples will be incorrect. Also, an important property like geometric ergodicity which are sufficient for CLT type results of ergodic averages are either not satisfied  by these algorithms \cite[see][for examples]{roberts1999} or some strong assumption on $\pi(\cdot)$ is needed. For example geometric ergodicity to hold for RWMH it is necessary (and sufficient) that the target density is \textit{log-concave in the tail} \citep{mengersen1996}. A density $\pi$ continuous and positive on \R ~is log-concave in the tail if there exists an $\a > 0$ and some $x_1 > 0$ such that 
\begin{equation}\label{eqn:logconcave}
\begin{split}
 y & \geq x \geq x_1 \quad \imply \quad\log\pi(x) - \log\pi(y)  \geq \a(y-x) \\
 y  & \leq x \leq -x_1  \quad \imply \quad\log\pi(x) - \log\pi(y)  \geq \a(x-y).
\end{split}
\end{equation}
and similarly
LMH is also geometrically ergodic under a strong assumption given in Theorem 4.1 of \cite{roberts1996}. These conditions are not satisfied for a large class of densities (e.g. the densities with thick-tails).

Thus a MH method which allows the proposed state to be far away from the current state and yet has good acceptance rate will be of much use in the statistical computing problems. It is even better if the algorithm is geometrically ergodic for a class of densities much larger than the classes for which this property is enjoyed by the standard algorithms. In this article we propose a new MH algorithm based on multiplying a random quantity with the states. Even if the algorithm appears simple we found that it has excellent convergence and mixing properties. It can explore the state space quite faster than the standard MCMC algorithms and has geometric ergodicity property for a huge class of target densities, for which the standard algorithms fails to be geometric ergodic. The main reason for this is because the dives can be made large or small each with significant probabilities. If the random multiplier is close to one, then the proposed point will be close to the current state and conversely. In RWMH, however, this proposal cannot be controlled easily. If the step size is chosen large then most of the proposed points would be far away from the current states and if the step size is chosen small then most of the proposed points would be very close to the current states.

There is obviously one issue with this algorithm -- the origin is an absorbing state. However, this is not vital since in major practical problems the variables are continuous and the origin has no mass. Thus we can safely remove it from the state space without disturbing the convergence. We emphasize that the RDMH algorithm exploits the multiplicative group structure of $\R - \{0\}$. The algorithm fails when the origin has positive probability attached to it (for example, when the state-space is the set of integers). The RWMH agorithm still work in that case. In other problems such as Bayesian testing with point null hypothesis and two-sided alternative, where the target distribution has a continuous part and also has a mass at zero, both RDMH and RWMH fails. 

MCMC techniques being extremely popular, the literature is rich with algorithms -- specialized or generic in nature. Many of them are special cases of MH algorithms with different forms  of proposal densities. We refer to \cite{liu2008} and \cite{robert2004} for book length discussions.

The structure of the article is as follows. We describe the new algorithm in section \ref{sec:rdmh}. As claimed already, the algorithm is \emph{new} in the sense that it is completely different in concept and in structure from the available multiplicative random walk MH. This is discussed in details in section \ref{sec:dissimilarity}. We discuss its convergence properties in section \ref{sec:convergence}. In section \ref{sec:appl:multimodal} we compare RDMH with the standard algorithms. Specifically we consider a bimodal target and see how RDMH explores the modes while RWMH cannot. We also consider an extreme mixture example where one of the components has very low dispersion compared to other. We then consider a thick tailed target for which RWMH is \emph{not} geometric ergodic while RDMH \emph{is}. In this example we see how asymptotic normality holds for the ergodic averages using RDMH while it fails to hold for the RWMH algorithm.In Section \ref{sec:appli:pricedata} we analyze a share price return data. Typically in such problems the posterior of one or more parameters are thick-tailed and asymmetric. The data and mode we consider are analyzed in \cite{fernandez1998} using Gibbs sampler. We found that the Gibbs sampler failed to explore the tail of the posterior of the location parameter while RDMH sampler did that with ease.  We conclude this article with an outlook on further works in Section \ref{sec:further}.

\section{Random dive MH}\label{sec:rdmh}
Suppose that $\pi$ is a target density function, probably unnormalized, on $\R$. At each iteration the algorithm proposes a state $y$ from the current state $x$ by multiplying a random quantity $\e$ with $x$ (i.e. $y = x\e$). The proposal is accepted with some probability depending on $x$ and $y$. We can classify the proposals into two classes depending on whether $|\e| \leq 1$ or $>1$. We call the case where $|\e| \leq 1$ an \emph{inner dive} and the case where $|\e| > 1$ an \emph{outer dive}. Notice that an outer dive can also be obtained by dividing the state $x$ by an $\e$ with $|\e| < 1$. Hence we can restrict the set from which the random multiplier $\e$ is drawn to the set  $\Y = (-1,1) - \{0\}$. Obviously the point zero is not considered so that the chain does not get stuck at zero. At each iteration we can take an inner dive or an outer dive at random. We call the chain symmetric if the probability for an inner dive is half and asymmetric otherwise. In this article we shall only consider the symmetric RDMH only. So with a proposal density $g(\e)$ for $\e$ on $\Y$,  the algorithm is given in Algorithm \ref{algo:unid}.
\begin{algo}\label{algo:unid} \topline Random dive MH on $\R$ \botline \normalfont \sffamily
\begin{itemize}
 \item Input: Initial value $x_0 \neq 0$, and number of iterations $N$. 
 \item For $t=0,\ldots,N-1$
\begin{enumerate}
 \item Generate $\e \sim g(\cdot)$ and $u \sim$ U$(0,1)$ independently
 \item If $0 < u < 1/2$, set 
\[ x' = x_t\e \quad \textrm{ and } \quad \alpha(x_t,\e) 
= \min\left\{\dfrac{\pi(x')}{\pi(x_t)} |\e|,~1 \right\}\]
\item Else set
\[ x' = x_t/\e \quad \textrm{ and } \quad \alpha(x_t,\e) 
= \min\left\{\dfrac{\pi(x')}{\pi(x_t)}\frac{1}{|\e|}, ~1 \right\}\]
\item Set \[x_{t+1} = \left\{\begin{array}{ccc}
 x' & \textsf{ with probability } & \a(x_t,\e) \\
 x_{t}& \textsf{ with probability } & 1 - \a(x_t,\e)
\end{array}\right.\]
\end{enumerate}
\item End for
\end{itemize}
\botline \rmfamily
\end{algo}

Notice that RDMH \emph{is} an MH algorithm with 
\begin{equation}\label{eqn:rdmh:proposal}
q(x \to y) = (1/2)~g(y/x)\dfrac{1}{|x|}\mathbb I(|y| < |x|) + (1/2)~g(x/y)\dfrac{|x|}{y^2}\mathbb I(|y| > |x|)
\end{equation}
Hence it follows that $\pi(\cdot)$ is indeed stationary for the chain. However other properties do not follow easily -- the assumptions in general results discussed in \cite{roberts2004} do not hold in this case. We prove them separately in the following section. Notice also that the terms $|\e|$ or its inverse in the acceptance ratios correspond to the Jacobian of the transformations : $x \mapsto x\e$ and $x\mapsto x/\e$ respectively. The acceptance ratios are free from the proposal density $g$ as in RWMH.

For each $x \neq 0$ we define the inner and outer acceptance regions respectively as,
\begin{eqnarray*}
 a(x) & = & \{ \e \in \Y ~:~ \pi(x\e)|\e|/\pi(x) \geq 1 \} \\
 A(x) & = & \{ \e \in \Y ~:~ \pi(x/\e)/(\pi(x) |\e| )\geq 1 \}
\end{eqnarray*}
Let $r(x) = \Y - a(x)$ and $R(x) = \Y - A(x)$ be the potential inner and outer rejection regions respectively.
We see that for each $x \ne 0$, the rejection probability is given by
\begin{eqnarray*}
 \rho(x) & = & \dint_{\R^{*}} (1 - \a(x\to y)) q(x\to y)~dy \\
 & = & \dfrac{1}{2}\dint_{|y| < |x|} (1 - \a(x\to y))\dfrac{1}{|x|}g\left(\dfrac yx\right)~dy + \dfrac{1}{2}\dint_{|y| \geq |x|} (1 - \a(x\to y))\dfrac{|x|}{y^2}g\left(\dfrac xy\right)~dy.
\end{eqnarray*}
Substituting $\e = y/x$ in the first integral and $\e = x/y$ in the second we get,
\begin{eqnarray*}
  \rho(x) & = & \dfrac{1}{2}\dint_{|\e| < 1} (1 - \a(x\to x\e))g(\e)~d\e + \dfrac{1}{2}\dint_{|\e| < 1} (1 - \a(x\to x/\e))g(\e)~d\e \\
 & = & \frac{1}{2}\dint_{r(x)} \left(1 - \dfrac{\pi(x\e)|\e|}{\pi(x)}\right)g(\e)d\e + \frac{1}{2} \dint_{R(x)} \left(1 - \dfrac{\pi(x/\e)}{\pi(x)|\e|}\right)g(\e)d\e
\end{eqnarray*}

Obviously, in any MH algorithm the proposal density plays an important role in terms of convergence. In this case also a good choice of $g$ is needed for faster convergence. However as we shall see in Theorem \ref{thm:geoerg} that the chain is geometric ergodic under an extremely weak restriction on $g$. Hence the discussion on the choices of $g$ is postponed till the end of Section \ref{sec:convergence}.

It is quite straightforward to extend the algorithm to higher dimensions. The variables may be updated either sequentially or jointly. While updating jointly at each iteration, outer dives should be applied to a random number of components (which may be zero) and inner dives to the rest. The Jacobian terms in the acceptance ratios will then be ratios of products of $\e$'s. The algorithm is given in Algorithm \ref{algo:highd} of Section \ref{sec:further}.

\section{Dissimilarity from RWMH}\label{sec:dissimilarity}
It is already seen that the RDMH algorithm \emph{is} a special case of MH class of algorithms. However, it is \emph{not} in any case similar to the random walk type algorithms. The term \emph{multiplicative random walk} (also known as the log-random walk MH) is not new in the statistics literature. However, it has been developed only when the state space is the positive half of the real line as $X_{t+1} = X_t\exp(N_t)$ where $N_t$ are i.i.d following some distributions on the real line (see, Dellaportas and Roberts, 2003, pp. 18 for details and Jasra \emph{et. al.}, 2005, for application). Obviously this reduces to the simple RWMH on $\R$ by observing that $\log X_{t+1} = \log X_t + N_t$. It is useless for problems having entire real line as support because a part of the state space is never visited (i.e the chain becomes \emph{reducible}), that is, positive (negative) initial values restrict the chain to take only positive (respectively negative) values.

The RDMH, however, is developed when the state sapce is entire real line. The term $\exp(N_t)$ above can not be equal to $\epsilon_t$ since $\epsilon_t$ can take both positive and negative values. Hence the RDMH algorithm cannot be considered as a special case of log-random walk MH. For distributions with $(0,\infty)$ as support the obvious way to emply RDMH is to reparametrize by taking logarithm so that the support becomes $\R$.

\section{Convergence Properties}\label{sec:convergence}
Let us denote the kernel the of the RDMH chain by $K(x \to y)$. Important properties like irreducibility and aperiodicity
\cite[see][for definitions]{roberts2004} are satisfied by the RDMH under minor assumption: 

\begin{theorem}\label{thm:ergodicity}
 If for all $0 < \delta <1$, $\inf_{\delta < |\e| < 1} g(\e) > 0$ and $\pi(\cdot)$ is bounded and positive on every compact subset of $\R$, then the chain is $\l$-irreducible (and hence $\pi$-irreducible) and aperiodic.
\end{theorem}

\begin{proof}
 Suppose $x\ne 0$ and $A \in \mathcal B(\R)$ has $\l(A) > 0$. Then there exists a compact set $C = \{u ~: ~r \leq |u| \leq R\}$, such that $\l(A\intersect C) > 0$ where $0 < r < |x| < R < \infty$. Define,
\[ I_x = [-|x|,~|x|] \qquad A^* = A\intersect C \qquad m = \inf_{y\in C} \pi(y) \quad \textrm{ and } \quad M = \sup_{y\in C}\pi(y)\]
and $a = \inf_{r/R < |\e| < 1}g(\e)$. The kernel $K$ of the chain satisfies,
\begin{eqnarray*}
K(x,A) & \geq & \dint_A q(y|x) \min\left\{ \frac{\pi(y)~q(y\to x)}{\pi(x)~q(x\to y)} ,1\right\}dy \\
& \geq & \dint_{A^*\intersect I_x} \frac{1}{2|x|}g(y/x)\min\left\{\frac{\pi(y)}{\pi(x)} \left|\frac{y}{x}\right|,1\right\}dy \\ &  & + \dint_{A^*\intersect I_x^c}\frac{|x|}{2y^2}g(x/y)\min\left\{ \frac{\pi(y)}{\pi(x)} \left|\frac{y}{x}\right|,1\right\}dy \\
& \geq & \dfrac{a}{2R}\min\left\{\dfrac{mr}{MR},1\right\} \l(A^*\cap I_x) + \dfrac{ra}{2R^2}\min\left\{\dfrac{mr}{MR},1\right\} \l(A^*\cap I_x^c) \\
& \geq & c~\l(A\cap C) > 0
\end{eqnarray*}
where $$c = \min\left\{\dfrac{a}{2R} , \dfrac{ra}{2R^2}\right\}\times\min\left\{\dfrac{mr}{MR},1\right\}.$$
Thus we see that, the chain is $\l$-\emph{irreducible}. Further since each measurable set with positive Lebesgue measure can be accessed in a single step implies the chain is \emph{aperiodic}.
\end{proof}
Thus the RDMH chain is \emph{ergodic} and hence by Theorem 4 of \cite{roberts2004},
\begin{equation}\label{eqn:tvconv}
 ||K^n(x,\cdot) - \pi(\cdot)||_{TV} ~\to ~0 \qquad \textrm{as } n \to \infty
\end{equation}
for $\pi-$almost every $x \in \R$. Where , $||\nu_1 - \nu_2||_{TV}$ is the well-known \textit{total variation distance} between two probability measures $\nu_1$ and $\nu_2$, defined as
\begin{equation}\label{eqn:tvnorm}
 ||\nu_1 - \nu_2||_{TV} ~ =~ \sup_A |\nu_1(A) - \nu_2(A)|
\end{equation}
For further properties of the total variation distance see \cite{meyn1993,roberts2004,robert2004} or \cite{liu2008}.
A sufficient condition for \eqref{eqn:tvconv} to hold for all $x$ rather than $\pi$--almost every $x$ is the \emph{Harris recurrence} of the chain. A Markov chain $(X_n)$ is called Harris recurrent if for every $x$ in the state space and for every set $B \in \mathcal B(\R)$ such that $\pi(B) > 0$, 
\[ P( \exists n : X_n \in B | X_0 = x] = 1.\]
Obviously the point \textit{zero} creates a problem in our RDMH algorithm. However, if we remove the single point zero from the state space, then since the chain is already $\pi$-irreducible we use Lemma 7.3 of \cite{robert2004} to conclude that
\begin{corr}
Under the assumptions of Theorem \ref{thm:ergodicity}, the RDMH chain is Harris recurrent on $\R^* = \R - \{0\}$.
\end{corr}
Thus \eqref{eqn:tvconv} holds for any nonzero $x$, i.e. any nonzero starting value ensures convergence of the chain.

A subset $C$ of $\statesp$ is called \emph{small} if there exists a positive integer $n$, a number $\delta > 0$ and a nontrivial measure $\nu$ such that
\begin{equation}\label{eqn:def:small}
 K^n(x,A) \quad\geq \quad\delta~\nu(A)\qquad \forall x \in C, \quad\forall A \in \mathcal{B}(\statesp)
\end{equation}
We will characterize the small sets for RDMH. In most of the MH algorithms any bounded subset of the state space is small. However this is not the case with RDMH. We first state a result for RDMH kernel useful in characterizing the small sets.
\begin{lemma}\label{lemma:weakconv}
 Suppose $(x_n)$ is a sequence of positive (negative) numbers decreasing (resp. increasing) to zero,  then $K(x_n,\cdot) \stackrel{w}{\longrightarrow} \delta_0$, where $\delta_0$ is the distribution degenerated at zero.
\end{lemma}
\begin{proof} Without loss we assume $(x_n) \downarrow 0$. Suppose $y < 0$. Then 
$$K(x_n,(-\infty,y]) ~\leq ~\dfrac{1}{2}\int_{x_n/y}^0 g(\e)d\e ~\to~ 0.$$ Also for $y > 0$, for sufficiently large $n$, $$K(x_n,(y,\infty) ~\leq ~\dfrac{1}{2}\int_{0}^{x_n/y} g(\e)d\e ~\to~ 0,$$ so that $K(x_n,(-\infty,y]) \to 1$. This completes the proof.
\end{proof}

\begin{theorem}\label{thm:smallset}
 Suppose the conditions in Theorem \ref{thm:ergodicity} holds. Then a set $E$ is small \emph{if and only if} its closure, $\bar E$ is a compact subset of $\R^* = \R-\{0\}$.
\end{theorem}

\begin{proof}
 Suppose first that $\bar E$ is compact subset of $\R^*$. Then $$r := \frac12\times\inf\{|x| : x \in E\} > 0 \quad \textrm{ and }\quad R = 2\times\sup\{|x| : x \in E\} < \infty.$$
So letting $C = \{u ~: ~ r \leq |u| \leq R\}$, it is seen from the proof of theorem \ref{thm:ergodicity} that,
\[K(x,A) \geq c~\l(A\cap C), \qquad \forall x \in E,~\forall A\in\mathcal B(\R) \]
Since $\l_C(A) := \l(A\cap C)$ is a nonzero measure on $(\R,\mathcal B(\R))$, this shows that $E$ is small.

Now suppose that $E$ is small. Clearly $\pm\infty$ cannot be limit points of $E$ since for any fixed $n$ and bounded $A$, $K^n(x,A) \to 0$ as $|x| \to \infty$. We shall also show that zero cannot be a limit point of $E$. This will show that $\bar E$ is compact subset of $\R^*$. So suppose on the contrary that zero is a limit point of $E$. Then there exists a sequence $(x_n)$ in  $E$ which monotonically converges to zero. 
Hence for any $m \in \mathbb{N}$ and any measurable set $A$, by Lemma \ref{lemma:weakconv}, 
\[ K^m(x_n,A) = \dint_{\R^*} K^{m-1}(y,A)K(x_n,dy) \longrightarrow \mathbb I(0 \in A) \quad \textrm{ as } n \to \infty.\]
So that \eqref{eqn:def:small} cannot hold for all $x\in E$ and all $A$ contradicting the assumption that $E$ is small.
\end{proof}

We now turn towards geometric ergodicity. An irreducible Markov kernel $K$ (irreducible with respect to some $\sigma-$finite measure $\nu$) with invariant distribution $\pi(\cdot)$ is said to be geometric ergodic if 
\[ \sup_{A \in \mathcal B(\R)}|K^n(x,A) - \pi(A)| \leq M(x)\rho^n, \quad\forall n\in \mathbb N\]
 for some $\rho < 1$, where $M(x) <\infty$, for $\pi-$a.e. $x \in \R$.

Geometric ergodicity is important in MCMC applications for the CLT of ergodic averages
\[ \hat{h}_N = \dfrac{1}{N}\sum_{i=1}^{N}h(X_i), \qquad N \in \mathbb N\]
of some function $h$ evaluated at each state of the Markov chain $(X_i)$. Corollary 2.1 of \cite{roberts1997} \cite[based on work of][]{kipnis1986} states that if a Markov kernel $P$ is geometric ergodic and reversible then for any function $h$ on the state space such that $\mathrm E_\pi{|h|^2} < \infty$
\[\sqrt{N}(\hat{h}_N - \pi(h)) \stackrel{w}{\longrightarrow} N(0,\sigma^2_h)\qquad \textrm{ as } N\to \infty\]
Such a CLT easily may not hold if the kernel is not geometric ergodic \cite[see][for examples]{roberts1999} or Section \ref{sec:appli:thick} of this article. Geometric ergodicity has multifarious usefulness discussed in \cite{jones2001} and \cite{roberts1998}.

To show that RDMH chain is geometrically ergodic, we put the following restriction on $\pi$. \\
\textbf{Assumption (A1).} For some $1 < p \leq \infty$
\begin{equation} \label{A1}
\lim_{|x|\to\infty} \pi(x)/\pi(x\e) ~= ~|\e|^p
\end{equation}
where the notation $|\e|^\infty$ should be interpreted as 0 for each $\e \in \Y$.
Notice that this is true for most of the posterior densities in Bayesian literature where MCMC finds extremely high applications. The condition (A1) given above is basically a regularly varying type restriction on $\pi$. Further discussion on regularly varying functions can be found in \citet[pp. 275--284]{feller1971}.

Recall that for each $x \ne 0$, the rejection probability is
\[\rho(x) = \frac{1}{2}\dint_{r(x)} \left(1 - \dfrac{\pi(x\e)|\e|}{\pi(x)}\right)g(\e)d\e + \frac{1}{2} \dint_{R(x)} \left(1 - \dfrac{\pi(x/\e)}{\pi(x)|\e|}\right)g(\e)d\e\]

We now give a bound on $\rho(x)$ in the following lemma.
\begin{lemma}\label{lemma:rhox}
 Assume (A1). Then 
\begin{eqnarray*}
 \rho(x) &\to & \frac{1}{2}\dint_\Y (1 - |\e|^{p-1})g(\e)d\e, \quad \textrm{ as } |x| \to \infty \\
 \rho(x) &\to & \frac{1}{2}\dint_\Y (1 - |\e|)g(\e)d\e, \quad \textrm{ as } x \to 0
\end{eqnarray*}
\end{lemma}
\begin{proof}
 Notice that as $|x|\to\infty$, $r(x) \to \phi,~ R(x) \to \Y$ and as $x\to 0$, $R(x) \to \phi, ~r(x)\to\Y$. Hence the result follows from dominated convergence theorem.
\end{proof}

Now we state a helpful result without proof.
\begin{lemma}\label{lemma:silly}
 Fix $p > 1$. For each $\e\in \Y$ and $s \in (0,1)$ define 
\begin{eqnarray*}
\varphi(s,\e) & = &  |\e|^s + |\e|^{1-s} - |\e| \\
\psi_p(s,\e) & = &  |\e|^{ps} + |\e|^{p-ps-1} - |\e|^{p-1}
\end{eqnarray*}
With $\psi_{\infty}(s,\e) \equiv 0$. Then
\begin{itemize}
 \item[(a)] $\varphi(s,\e) < 1$ for all $\e \in \Y$ and $s \in (0,1)$.
\item[(b)] $\psi_p(s,\e) < 1$ for all $\e \in \Y$ and $0 < s < 1/2 - 1/(2p)$.
\end{itemize}
\end{lemma}
As discussed in Theorem 15.0.1 of \cite{meyn1993}, the RDMH chain is geometric ergodic if and only if for some small set $E$, some function $V : \R^* \to [1,\infty)$ which is finite at least for one $x$, $\gamma < 1$ and some $b_E < \infty$, the \emph{geometric drift condition} holds:
\begin{equation}\label{eqn:geodrift}
 KV(x) \leq \gamma V(x) + b_E\mathbb{I}(x\in E),
\end{equation}
where $KV(x) = \dint_{\R^*}K(x\to y)V(y)dy$. In our case, we know any compact set of $\R^*$ is small. So if we can show for some continuous $V : \R^* \to [1,\infty)$ which is bounded on every compact subset of $\R^*$, the following conditions hold:
\begin{equation}\label{eqn:limsupcond}
 \limsup_{|x| \to \infty } \dfrac{KV(x)}{V(x)} < 1 \quad\textrm{ and } \quad\limsup_{x\to 0 } \dfrac{KV(x)}{V(x)} < 1
\end{equation}
then we can choose a number $\gamma < 1$ and a small set $E = \{ x : r \leq |x| \leq R \}$ for some $0 < r < R < \infty$, such that $KV(x) < \gamma V(x)$ for all $x \notin E$. Also $b_E = \sup_{x\in E} K V(x) < \infty$ since $V$ is bounded on $E$. Hence we see that \eqref{eqn:geodrift} holds.

We now state and prove the most important theorem of this section.
\begin{theorem}\label{thm:geoerg}
 Suppose that conditions in Theorem \ref{thm:ergodicity} holds together with continuity of $\pi(x)$ and (A1). Further assume the following: for some $s_0 \in (0,1)$, 
\begin{equation}\label{eqn:regularity_g}
 \dint_{-1}^{1} |\e|^{-s_0}g(\e)~d\e ~<~\infty
\end{equation}
Then the chain is geometrically ergodic.
\end{theorem}
\begin{proof}
 In view of discussion preceding the statement of the theorem we only need to show \eqref{eqn:limsupcond}. Fix $0 < s < \min\{s_0,1/2-1/(2p)\}$. Then by \eqref{eqn:regularity_g}
\begin{equation}\label{eqn:regulariry_g2}
 \dint_{-1}^{1} |\e|^{-s}g(\e)~d\e ~<~\infty
\end{equation}
Now Notice that, for each $x \ne 0$, and any function $V:\R^* \to [1,\infty)$
\begin{equation}\label{eqn:kernratio}
\begin{split}
 \dfrac{KV(x)}{V(x)} & = \frac{1}{2}\dint_{a(x)}g(\e)\frac{V(x\e)}{V(x)}d\e + \frac{1}{2}\dint_{A(x)}g(\e)\frac{V(x/\e)}{V(x)}d\e \\
& +  \frac{1}{2}\dint_{r(x)}g(\e)\frac{\pi(x\e)}{\pi(x)}|\e|\frac{V(x\e)}{V(x)}d\e + \frac{1}{2}\dint_{R(x)}g(\e)\frac{\pi(x/\e)}{\pi(x)|\e|}\frac{V(x/\e)}{V(x)}d\e + \rho(x)
\end{split}
\end{equation}
Choose positive constants $c_1$, $c^*_2$ and $c^{**}_2$ such that the function
\[ V(x) = c_1\pi(x)^{-s}\mathbb{I}(|x| > 1) + c^*_2 x^{-s}\mathbb{I}(0 < x \leq 1) + c^{**}_2 (-x)^{-s}\mathbb{I}(-1 \leq x < 0)\]
is continuous\footnote{The following choices do the job: $\sup_{|x| > 1} \pi(x)^s < c_1 < \infty$, $c_2^* = c_1\pi(1)^{-s}$ and $c_2^{**} = c_1\pi(-1)^{-s}$} and $V(x) \geq 1$ for all $x \ne 0$ and let $c_2 = \min\{c_2^*,~c_2^{**}\}$ and $C_2 = \max\{c_2^*,~c_2^{**}\}.$ We now work with this $V$ to show \eqref{eqn:limsupcond}

\underline{\emph{Case I}:} Suppose $|x| \to\infty$.  Then assumption (A1) implies that $A(x) \to \phi$ and $r(x) \to \phi$. Notice in this case
\[ \dfrac{V(x/\e)}{V(x)} ~= ~ \left(\dfrac{\pi(x)}{\pi(x/\e)}\right)^s ~\leq~ |\e|^{-s} \qquad\forall ~\e \in A(x)\]
Hence by \eqref{eqn:regulariry_g2} the second integral in \eqref{eqn:kernratio} converges to zero. Also, since $\forall ~\e \in r(x)$
\[ \begin{split}
\dfrac{\pi(x\e)}{\pi(x)}|\e|\dfrac{V(x\e)}{V(x)} & \leq \max\left\{ \left(\dfrac{\pi(x\e)}{\pi(x)}\right)^{1-s}|\e|,~\dfrac{\pi(x\e)}{\pi(x)^{1-s}}|\e||x|^{-s}|\e|^{-s}\dfrac{c_1}{c_2}\right\}\\
& \leq \max\left\{|\e|^s, M^s\dfrac{c_1}{c_2}\right\}\end{split},\]
 where $M = \sup\pi(x)$, the third integral in \eqref{eqn:kernratio} also converges to zero. Further since $V(x\e)/V(x) \to |\e|^{ps}$
it follows from \eqref{eqn:kernratio} and Lemma \ref{lemma:rhox} that 
\begin{eqnarray*}
 \limsup_{|x| \to \infty} \dfrac{KV(x)}{V(x)} & \leq & (1/2)\dint_{-1}^1|\e|^{ps}g(\e)d\e + (1/2)\dint_{-1}^1|\e|^{p-ps-1}g(\e)d\e  \\
 & &  + (1/2)\dint_{-1}^1(1-|\e|^{p-1})g(\e)d\e \\
& = & \frac{1}{2}\dint_{-1}^1\psi_p(s,\e)g(\e)d\e + \frac{1}{2}\\
& < & 1 \qquad \textrm{ by Lemma \ref{lemma:silly}}
\end{eqnarray*}

\underline{\emph{Case II}:} Suppose now that $x \to 0$. In this case $a(x) \to \phi$ and $\R(x)\to\phi$. Also $V(x\e)/V(x) \leq \frac{C_2}{c_2}|\e|^{-s}$ for all $|x| < 1$ implies that the first integral in \eqref{eqn:kernratio} converges to zero by \eqref{eqn:regulariry_g2}. Also  $\forall~ |x| < 1$ and $\forall~\e \in R(x)$
\[ \dfrac{\pi(x/\e)}{\pi(x)|\e|}\dfrac{V(x/\e)}{V(x)} \leq \max\left\{ |\e|^s,\dfrac{\pi(x/\e)^{1-s}}{\pi(x)}|x|^s|\e|^{s-1}\dfrac{C_2}{c_1} \right\} \leq \max\left\{ |\e|^s,m^{-s}\dfrac{C_2}{c_1}\right\},\]
where $m = \inf_{|x| < 1} \pi(x) > 0$, so that the fourth integral in \eqref{eqn:kernratio} converges to zero. Hence by continuity of $\pi(\cdot)$ at zero, 
\begin{eqnarray*}
 \limsup_{x\to 0}\dfrac{KV(x)}{V(x)} & \leq & \frac{1}{2}\dint_{-1}^1|\e|^{s}g(\e)d\e + \frac{1}{2}\dint_{-1}^1|\e|^{1-s}g(\e)d\e + \frac{1}{2}\dint_{-1}^1(1-|\e|)g(\e)d\e\\
& = & \frac{1}{2}\dint_{-1}^1\varphi(s,\e)g(\e)d\e + \frac12 \\
& < & 1 \qquad \textrm{ by Lemma \ref{lemma:silly}}
\end{eqnarray*}
This completes the proof.
\end{proof}

\textbf{Remark:} It is conjectured in \cite{atchade2007} that an MH chain is geometric ergodic if the rejection probability is bounded away from 1. They have proved the result with an additional assumption that the continuous part of the MH kernel, i.e. \mbox{$\a(x\to y)q(x\to y)dy$} which is an operator on $L^2(\pi)$, is compact. Unfortunately this extra assumption does not hold for RDMH (along with most of the MH algorithms). Had the conjecture been proved, we could have claimed readily that RDMH is geometrically ergodic by using Lemma \ref{lemma:rhox}. This would not require the extra assumption \eqref{eqn:regularity_g} on $g$.

The class of densities satisfying (A1) together with the assumptions of Theorem 1 and continuity is quite large. This class obviously includes the following classes:
\begin{enumerate}
 \item The class of thick-tailed densities $\pi(x) \sim \dfrac{1}{p(x)^m}$ as $|x| \to \infty$ where $p(x)$ is a polynomial satisfying $p(x) > 0$ for all sufficiently large $|x|$. For example, the $t$-densities fall in this class.
\item The class of densities which are equally log-concave in the two tails, i.e., for some $M > 0$ and some $\a > 0$,
\begin{equation}\label{eqn:stronglogconcave}
 |y| > |x| > M \quad \imply \quad \log\pi(x) - \log\pi(y) \geq \a (|y| - |x|)
\end{equation}
This is a stronger version of \eqref{eqn:logconcave} for \eqref{eqn:stronglogconcave} implies \eqref{eqn:logconcave}. Notice that for these densities, $p = \infty$ in (A1). Examples of such densities are the normal densities and their mixtures, double exponential density etc. 

\item The class of densities of the form $h(x)\exp(-\kappa\sqrt[m]{|x-\theta|})$ where $m > 1$.
\end{enumerate}

In some problems, however, the target $\pi(x)$ is log-concave in the tails but the rates at which $\pi(x)$ converges to zero are not same for the two tails. One example of such densities is  $f(x) = \exp(x - \exp(x))$. Notice that if $Y$ follows the standard exponential distribution $\mathcal{E}xp(1)$ then $X = \log Y$ has density $f(x)$. It can be seen that $f(x)/f(x\e) \to 0$ holds for each $\e > 0$, as $|x| \to \infty$ and also for each $\e < 0$ and as $x \to \infty$. But $f(x)/f(x\e) \to \infty$ if $\e < 0$ and $x\to -\infty$. We assume for these kind of densities exactly one tail dominates, i.e. exactly one the following is true for each $\e \in (-1,0)$.
\begin{eqnarray}
\lim_{x\to-\infty} \pi(x)/\pi(x\e) \to \infty \quad \textrm{and} \quad \lim_{x\to\infty} \pi(x)/\pi(x\e) \to 0 \label{eqn:A2} \\
\lim_{x\to-\infty} \pi(x)/\pi(x\e) \to 0 \quad \textrm{and} \quad \lim_{x\to\infty} \pi(x)/\pi(x\e) \to \infty \label{eqn:A2'}
\end{eqnarray}
Notice that it is sufficient to work with \eqref{eqn:A2} because if \eqref{eqn:A2'} holds for a target $\pi$, then $\pi(-x)$ satisfies \eqref{eqn:A2}.
For these kind of densities (A1) does not hold. However, the next theorem assures that the RDMH chain is still geometric ergodic. The proof is along the line of Theorem \ref{thm:geoerg} and so we just present a sketch.

\begin{theorem}
 Suppose $\pi$ is continuous and the assumptions in Theorem \ref{thm:ergodicity} holds. Suppose further that $\pi(x)$ satisfies \eqref{eqn:A2} and that $g$ satisfies the regularity condition \eqref{eqn:regularity_g}. Then the RDMH chain is geometric ergodic.
\end{theorem}
\begin{proof}
 Notice that in this case we have the following as $x \to \infty$ we still have $a(x) \to \Y$ and $r(x) \to \phi$. But 
\[ A(x) \cap (0,1) \to \phi \qquad A(x) \cap (-1,0) \to (-1,0)\]
and
\[ R(x) \cap (0,1) \to (0,1) \qquad R(x) \cap (-1,0) \to \phi.\]
Thus in \ref{eqn:kernratio} (with the same choice of $V$ as in theorem \ref{thm:geoerg}) we can further split the integrals on intersections of the domains with $(-1,0)$ and $(0,1)$. On each such domain either the integrand converges to zero or it is bounded and the domain of the integral converges to the empty set. Hence
\[ \limsup_{x\to\infty}\dfrac{KV(x)}{V(x)} \leq \limsup_{x\to\infty} \rho(x) <  1/2\]
since,
\begin{eqnarray*}
 2\rho(x) & = & \dint_{r(x)\cap (0,1)} \left(1 - \dfrac{\pi(x\e)|\e|}{\pi(x)} \right)g(\e)d\e + \dint_{r(x)\cap (-1,0)} \left(1 - \dfrac{\pi(x\e)|\e|}{\pi(x)} \right)g(\e)d\e \\
&  & + \dint_{R(x)\cap (0,1)} \left(1 - \dfrac{\pi(x/\e)}{\pi(x)|\e|} \right)g(\e)d\e + \dint_{R(x)\cap (-1,0)} \left(1 - \dfrac{\pi(x/\e)}{\pi(x)|\e|} \right)g(\e)d\e \\
& \to & \dint_0^1 g(\e)d\e < 1 \quad\textrm{ as }x \to \infty
\end{eqnarray*}
Also as $x \to -\infty$, it can be seen that $A(x) \to \phi$ and $R(x) \to \Y$ hold but
\[ a(x)\cap(0,1) \to (0,1), \qquad a(x)\cap (-1,0) \to \phi \]
and
\[ r(x)\cap(0,1) \to \phi \qquad r(x) \cap (-1,0) \to (-1,0)\]
Hence similarly,
\[ \limsup_{x\to\infty}\dfrac{KV(x)}{V(x)} \leq \limsup_{x\to\infty} \rho(x) <  1\]
since in this case
\begin{eqnarray*}
 2\rho(x) & = & \dint_{r(x)\cap (0,1)} \left(1 - \dfrac{\pi(x\e)|\e|}{\pi(x)} \right)g(\e)d\e + \dint_{r(x)\cap (-1,0)} \left(1 - \dfrac{\pi(x\e)|\e|}{\pi(x)} \right)g(\e)d\e \\
&  & + \dint_{R(x)\cap (0,1)} \left(1 - \dfrac{\pi(x/\e)}{\pi(x)|\e|} \right)g(\e)d\e + \dint_{R(x)\cap (-1,0)} \left(1 - \dfrac{\pi(x/\e)}{\pi(x)|\e|} \right)g(\e)d\e \\
& \to & 2\dint_{-1}^0 g(\e)d\e +  \dint_0^1 g(\e)d\e < 2 \quad\textrm{ as }x \to -\infty
\end{eqnarray*}
This verifies the first condition of \eqref{eqn:limsupcond}. Verification of the second condition of \eqref{eqn:limsupcond} is already done in \emph{case II} of Theorem \ref{thm:geoerg}.
\end{proof}

We now return to the choices of $g$. Let $\mathcal{B}(x;a,b)$ denote the density of a Beta$(a,b)$ random variable. A general class of proposal densities satisfying \eqref{eqn:regularity_g} is then given by
\[  g(\e) = \g\B(-\e;a_1,b_1) \mathbb{I}(-1 < \e < 0) + (1-\g)\B(\e;a_2,b_2)\mathbb{I}(0 < \e < 1).\]
for some $0 < \g < 1$ and some positive numbers $a_1,a_2,b_1,b_2$. Notice that a straightforward choice of $g$ is uniform distribution over (-1,1) which corresponds to the case  $\gamma = 1/2$ and $a_1 = a_2 = b_1 = b_2 = 1$. This indeed allows for large dives but the acceptance rate may drop. Further if the simulated $\e$ is very close to zero and a outer dive is taken, then the proposed state will have large magnitude and result in numerical instability. Specially when the posterior is highly steep then such large dives are not sensible. We shall nevertheless use this choice of $g$ in the next section and show it works well. If however, the target density is steep then a good idea would be to generate $\e$'s close to 1. This can be achived by
\begin{enumerate}
 \item making $\g$ small (e.g. $\g \leq 0.2$) and
 \item making $a_2$ high and $b_2$ small (e.g. $a_2 \geq 2$ and $0 < b_2 \leq 1$).
\end{enumerate}

\section{Application}
In this section we consider two simulation studies. We first consider a bimodal target density and show how the RDMH algorithm explores the modes but the RWMH chain either gets stuck at the mode (if the proposal variance is moderate) or explores the modes at high value of proposal variance but has very low acceptance rate. In the next example we consider another simulation study on a thick tailed target. The RWMH and the LMH algorithms are \emph{not} geometrically ergodic for this target under any kind of proposal (thick-tailed or thin tailed) and this has a serious effect when we try to construct a confidence set based on asymptotic normality of ergodic averages -- for the latter does not hold in this case. However the RDMH is still geometric ergodic and a CLT holds for the ergodic averages.
\subsection{Exploring a multimodal target}\label{sec:appl:multimodal}
\paragraph{Example 1.} Consider the mixture distribution
\[ \pi(x) = 0.5~\phi(x;0,0.25) + 0.5~\phi(x;10,0.25)\]
where $\phi(x;\mu,\sigma) = \exp\left(-0.5(x-\mu)^2/\sigma^2)\right)/(\sigma\sqrt{2\pi})$ is the normal 
density with mean $\mu$ and variance $\sigma^2$.

Clearly this is a bimodal distribution with two separated modes at $x = 0$ and $x = 10$. We compare the RDMH with the RWMH here. We choose $g$, as the uniform distribution on $\Y = (-1,1) - \{0\}$ i.e.
\[ g(\e) = 1/2 \qquad -1 < \e < 1, ~\e \ne 0 \]

\begin{figure}
\centering
 \includegraphics[width=0.9\textwidth]{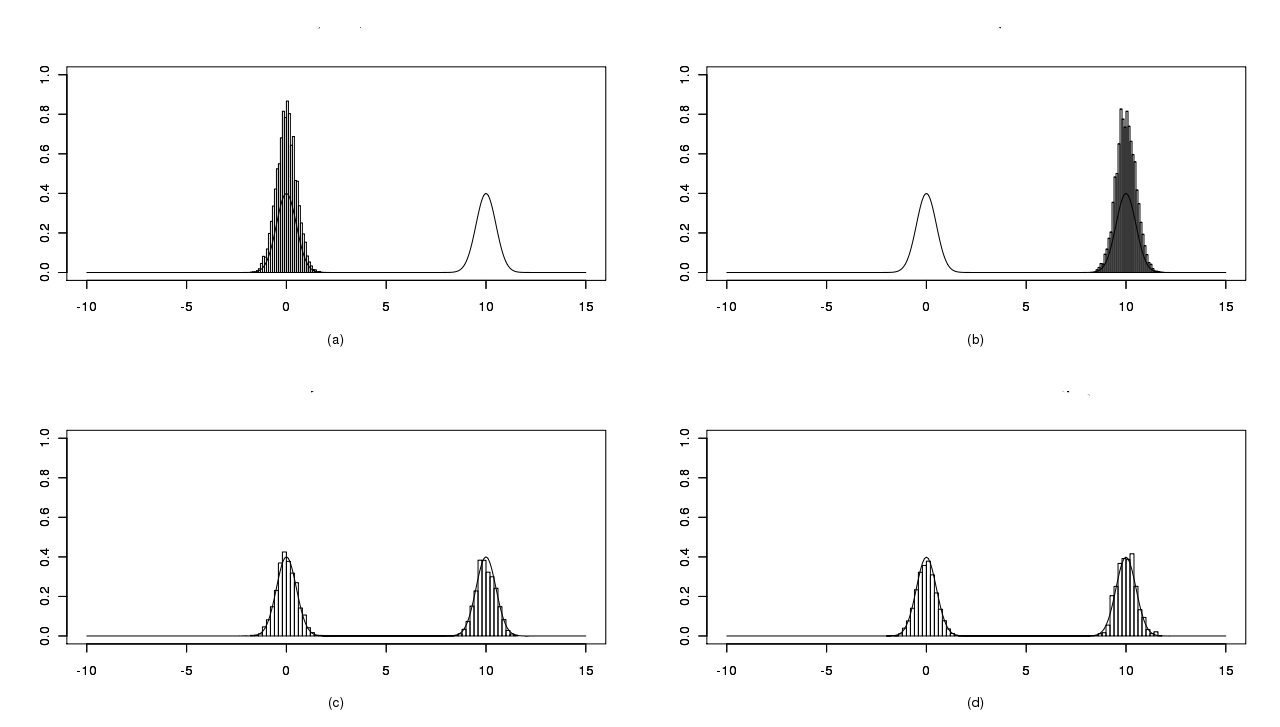}
 \caption{{\small Histograms and True density(solid curve) based on 30,000 sample values out of 50,000 sample values  (burn-in = 20,000) for the bimodal example in section \ref{sec:appl:multimodal}}} \flushleft
\begin{small}
(a) RWMH: $\tau = 2$, initial value: $x_0 = 0$ acceptance rate = 29.916\%, (b) RWMH: $\tau = 2$, initial value: $x_0 = 10$ acceptance rate = 29.656\%, (c)  RWMH: $\tau = 5$, initial value: $x_0 = 10$ acceptance rate = 14.31\%, (d) RDMH, initial value: $x_0 = -2$, acceptance rate = 30.172\% 
\end{small}
 \label{fig:bimodal}
\end{figure}

For the RWMH, we choose the proposal $q(x'|x) = \phi(x';x,\tau^2)$ for different choices of $\tau$. 
Figure  \ref{fig:bimodal} (panels (a) and (b)) shows that RWMH remains stuck at one of the modes for an arbitrary but reasonable choice of $\tau$, and with arbitrary initial value. This indicates significant non-robustness of RWMH with respect to the initial value and the choice of $\tau$ even if it is geometric ergodic in this case. Only when $\tau$ has been appropriately chosen, RWMH performs adequately (Figure \ref{fig:bimodal}, panel (c)). We remark that such ``right" choice is possible only if bimodality of the target posterior is anticipated beforehand, which is unrealistic. Even for the appropriate choice of $\tau$ we notice that the acceptance rate of RWMH is rather small (14.31\%). In contrast RDMH adequately explored the entire state space without requiring knowledge of the target density (Figure \ref{fig:bimodal}, panel (d)), or tuning of the proposal. The acceptance rate, which is 30.172\%, much encouraging compared to the RWMH algorithm.

\paragraph{Example 2.} Now we consider a more challenging case similar to one considered by \citet[p. 1632--1634]{chen_kim2006} in their Example 2. The target is a mixture of univariate normals:
\[ \pi(x) = 0.5~\phi(x;0,10^{-4}) + 0.5~\phi(x;5,1)\]
Several specialized algorithms are available for such \emph{needle-in-haystack} problem among which the equi-energy sampler by \cite{kou2006} is worth mentioning. The equi-energy sampler is extremely efficient once the tuning parameters are chosen carefully. For our purpose we chose few proposals at random and study their performances. In particular, we run each chain (of length 30,000 each after discarding first 20,000 burn-ins) 100 times and estimate $$\hat{P} = \frac{1}{30000}\sum_{i=1}^{30000}\mathbb{I}( |X_t| <0.05)$$, where $X_t$ denotes the Markov chain. From Table \ref{tab:needle_example_table} it can be seen that all the proposals work quite well. The third and fifth proposal results small M.S.E's perhaps due to the fact that they put more weight near zero (the multiplier is close to zero and hence so is the proposed state) and one of the mode is at zero. However, since in practice, it need not be the case it might result in poor acceptance rates. So, a proposal that generates random multiplier  close to 1 should be preferred.

\begin{table}[!h]
 \begin{center}
  \begin{tabular}{|c|cccc|}\hline
   Proposal & $E(\hat{P})$ & s.d.($\hat P$) & MSE($\hat P$) & Avg. accep. rates \\ \hline
   U(-1,1) & 0.5115 & 0.1361 & 0.0184 & 37.14\% \\
  $0.15 \B(-\e;1,1) + 0.85 \B(\e;1,1)$ & 0.4953 & 0.1218 & 0.0147 & 41.08\%\\
  $0.15 \B(-\e;0.5,1) + 0.85 \B(\e;0.5,1)$ & 0.4938 & 0.0470 & 0.0022 & 28.51\% \\
  $0.15 \B(-\e;1,0.5) + 0.85 \B(\e;1,0.5)$ & 0.4912 & 0.1767 &0.0310 & 56.79\%\\
  $0.5 \B(-\e;0.5,0.5) + 0.5 \B(\e;0.5,0.5)$ & 0.5024 & 0.0520 & 0.0027 & 37.91\% \\ \hline
  \end{tabular}
 \end{center}
\caption{Mean, standard deviation and mean squared errors of $\hat{P}$ for different proposals. The average acceptance rates are also reported.}
\label{tab:needle_example_table}
\end{table}

\subsection{Exploring a thick-tailed target}\label{sec:appli:thick}
In this section we consider a thick-tailed target density $\pi$ for which RWMH and LMH are not geometrically ergodic but RDMH is. We chose
\begin{equation}\label{eqn:thicktarget}
 \pi(x) = \dfrac{2}{\pi}\dfrac{1}{(1+x^2)^2}
\end{equation}
It is easy to verify that $\mathrm{E}_\pi|X|^2 < \infty$. Any other thick tailed density or a density which is not log-concave in the tail could have been chosen in place of \eqref{eqn:thicktarget}. 

Notice that $\nabla\log\pi(x) \to 0$ as $|x| \to \infty$. So $\pi$ cannot be log-concave in the tail and hence the RWMH chain is not geometrically ergodic. Moreover, Theorem 4.3 of \cite{roberts1996} assures that the LMH is not geometrically ergodic either.
\begin{table}[!h]
\centering
\caption{P-values of tests for normality performed on the samples of means.}
 \begin{tabularx}{\textwidth}{|p{4cm}|Z|Z|ZZZ|}
 \hline
 & Acceptance & & \multicolumn{3}{c|}{P-value of} \\ 
Algorithm  & rate (\%) & s.e. & \textsf{AD} & \textsf{CVM} & \textsf{Lill} \\ \hline
RDMH 	    & 66.43	 & 0.0074	 & 0.8242    & 0.8241	 & 	0.5737 \\ \hline
RWMH $(\mathcal{N}(0,1.5^2))$ & 46.99	 & 0.0298	 & 	 & 		 & 	 \\
RWMH $(\mathcal{C}(0,1))$ & 44.95	 & 0.0182	 & 	 & 		 & 	 \\
LMH (scale = 2)  & 87.90	 & 0.1767	 	& 	 & 	0	 & 	 \\
LMH (scale = 3)  & 79.39	 & 0.5460	 & 	 & 		 & 		\\
LMH (scale = 4)  & 78.17	 & 0.8631	 & 	 & 		 & 		 \\ \hline 
\end{tabularx}

\label{tab:thick_tests}
\end{table}

Our parameter of interest was $\mathrm{E}_\pi X$. We compared the RDMH algorithm with the RWMH and the LMH algorithms. For RDMH proposal $g$ we again chose the uniform distribution over $(-1,1)$. For RWMH, however we chose two proposals -- one thin-tailed and one thick-tailed. The the thin-tailed proposal is a normal distribution with mean zero and variance $1.15^2 ~(\mathcal{N}(0,1.5^2)$ while thick-tailed proposal is the standard Cauchy distribution $(\mathcal{C}(0,1))$. The scale paramters for the LMH were chosen to be 2, 3 and 4. For each of the algorithms, we ran 1000 independent chains of lengths 50,000 each and obtained the means of last 40,000 values of each such chain. Thus we obtained six samples of estimates of $\mathrm{E}_\pi X$ each of which had size 1000.
Three tests were performed on each of these three samples in order to quantitatively assess the normality behavior. The three tests were the Anderson--Darling (AD) test, Cramer--von Mises (CVM) test and the Lilliefors (Lill) test for normality. The descriptions of the tests can be found in \cite{thode2002}. The tests were performed by \textsf{nortest} package of \textsf{R}-statistical software. The p-values together with the average acceptance rate of each of the 1000 chains and standard error of each of the six samples of empirical means are reported in Table \ref{tab:thick_tests}. The QQ--plots are shown in Figure \ref{fig:thick:qqplots} and the auto-correlation plots for a typical run of the samplers are shown in Figure \ref{fig:thick:acfplots} (no thinning).

It is seen both from the p-values and the QQ--plots that normality holds for the empirical means obtained by RDMH algorithm while in the RWMH and the LMH algorithms they are far from normality.
\begin{figure}[tph]
\centering\includegraphics[width=\textwidth]{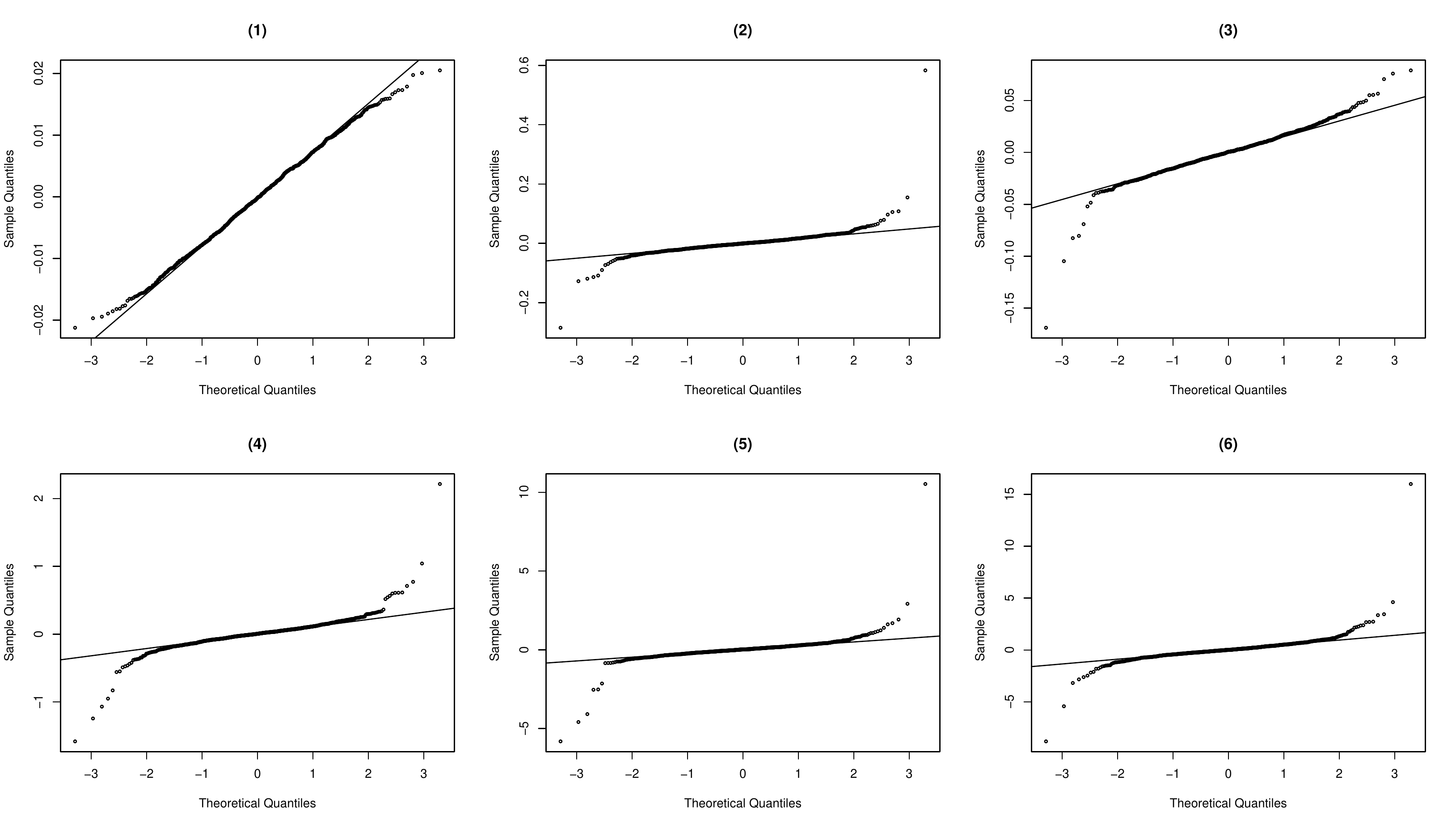}
\caption{QQ--plots of six samples of empirical means. (1) : RDMH; (2) : RWMH with $\mathcal{N}(0,1.5^2)$ proposal; (3) RWMH with $\mathcal{C}(0,1)$ proposal; (4) - (6) : LMH with scales 2, 3 and 4 respectively}
\label{fig:thick:qqplots}
\centering\includegraphics[width=\textwidth]{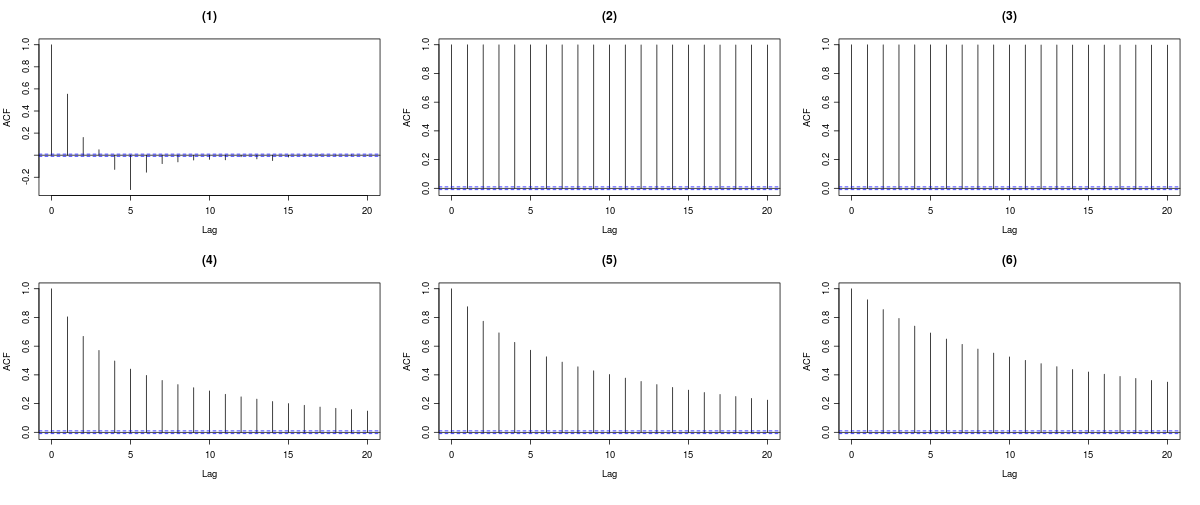}
\caption{Auto-correlation plots of six samplers. (1) : RDMH; (2) : RWMH with $\mathcal{N}(0,1.5^2)$ proposal; (3) RWMH with $\mathcal{C}(0,1)$ proposal; (4) - (6) : LMH with scales 2, 3 and 4 respectively}
\label{fig:thick:acfplots}

\end{figure}

Next, to judge how fast RDMH converges in this scenario we conducted a further study. For each of the algorithms we ran thousand independent chains of length 1000 each and calculated $d(\hat F,F_0)$ -- the Kolmogorov--Smirnov distance between the empirical c.d.f $(\hat F)$ and the true c.d.f.
\[ F_0(x) = \frac{1}{\pi}\arctan x + \frac{1}{2} + \frac{1}{2\pi}\sin(2\arctan x) \]
by the formula
\[ d(\hat F,F_0) = \sup_{x\in \R} |\hat F(x) - F_0(x)|\]
and also obtained the p-values for testing $H_0 : \hat F = F_0$ against the two sided alternative. Table \ref{tab:ksdistances} reports the averages of the Kolmogorov-Smirnov distances and the p-values from the 1000 independent chains. This shows that the RDMH chain converges much faster compared to the RWMH and the LMH algorithms.

\begin{table}
\centering
\caption{Performances of the algorithms in terms of Kolmogorov-Smirnov distances.}
\label{tab:ksdistances}
\begin{tabular}{|c|c|c|c|c|c|c|c} \hline
 Method: & RDMH & \multicolumn{2}{c|}{RWMH}  & \multicolumn{3}{c|}{LMH} \\ \cline{3-7}
	 & 	& $\mathcal{N}(0,1.5^2)$ & $\mathcal{C}(0,1)$ & Scale = 2 & Scale = 3 & Scale = 4 \\ \hline
KS value & 0.0202 & 0.0788 & 0.0647 & 0.3682	& 0.4665 & 0.5000 \\
p-value  & 0.7997 & 0.0338 & 0.0365 & 0	& 0 & 0 \\ \hline
\end{tabular}
 
\end{table}

\subsection{Share price return data}\label{sec:appli:pricedata}
 In this section we consider the daily price returns of Abbey National share between July 31 and October 8, 1991. The data is presented in Table 1 of \cite{buckle1995}. We consider the simple location-scale model proposed and analyzed in \cite{fernandez1998}. Let $p_i, i=0,\ldots,49$ denote the price data in Table 1 of \cite{buckle1995} and $y_i = (p_i - p_{i-1})/p_{i-1},~i=1,\ldots,49$. \cite{fernandez1998} modeled the data as follows:
\begin{equation}\label{eqn:pricemodel}
\begin{split}
 p(y_1,\ldots,y_n | \b, \s, \nu, \g) & = \left[ \dfrac{2}{\g+\frac{1}{\g}}\quad \times\quad \dfrac{\Gamma((\nu+1)/2)}{\Gamma(\nu/2)\sqrt{\pi\nu}}\quad\s^{-1}\right]^n \quad \times \\
& \qquad \prod_{i=1}^n\left[ 1 + \dfrac{(y_i - \b)^2}{\nu\s^2}\left\{\frac{1}{\g^2} \mathbb{I}(y_i > \b) + \g^2\mathbb{I}(y_i < \b)\right\}\right]^{-(\nu+1)/2}\\
\end{split}
\end{equation}
with independent priors on the parameters as follows:
\begin{eqnarray*}
 p(\b) =  1 & ; & p(\s) = 1/\s\\
 p(\nu) = d\exp(-d\nu) &;& p(\phi) = b^a\Gamma(a)^{-1}\phi^{a-1}\exp(-b\phi)
\end{eqnarray*}
where $\phi = \g^2$. The hyper-parameters are given by \cite{fernandez1998} as $d = 0.1, a = 1/2$ and $b = 1/\pi$. We log-transform all the parameters except $\b$ so that the state space becomes $\mathbb{R}^4$. That is, we re-parametrize: $\tilde\sigma = \log\s,~\tilde\nu = \log \nu$ and $\tilde \g = \log \g$. We updated the parameters $(\b,~\tilde\s,~\tilde\nu,\tilde\g)$ sequentially with the following  proposal densities:
\begin{eqnarray*}
 g_\b(\e) & = & 0.80 \mathcal{B}(\e;2,1)\mathbb{I}(0 < \e < 1) ~+ ~0.20 \mathcal{B}(-\e;2,1)\mathbb{I}(-1 < \e < 0) \\
 g_{\tilde\s}(\e) & = & 0.80 \mathcal{B}(\e;3,0.5)\mathbb{I}(0 < \e < 1) ~+ ~0.20 \mathcal{B}(-\e;3,0.5)\mathbb{I}(-1 < \e < 0) \\
 g_{\tilde\nu}(\e) & = & 0.80 \mathcal{B}(\e;3,0.5)\mathbb{I}(0 < \e < 1) ~+ ~0.20 \mathcal{B}(-\e;3,0.5)\mathbb{I}(-1 < \e < 0) \\
 g_{\tilde\g}(\e) & = & 0.80 \mathcal{B}(\e;2,0.5)\mathbb{I}(0 < \e < 1) ~+~ 0.20 \mathcal{B}(-\e;2,0.5)\mathbb{I}(-1 < \e < 0)
\end{eqnarray*}
where $\mathcal{B}(\e;a,b)$ is the Beta density proportional to $ \e^{a-1}(1-\e)^{b-1}\mathbb{I}(0 < \e < 1)$. We tried couple of other such mixtures too and the results were very close. Using a proposal density uniform on $(-1,1)$ is not a good idea in this case as discussed before.

\begin{figure}[htp]
 \centering\includegraphics[width=0.8\textwidth]{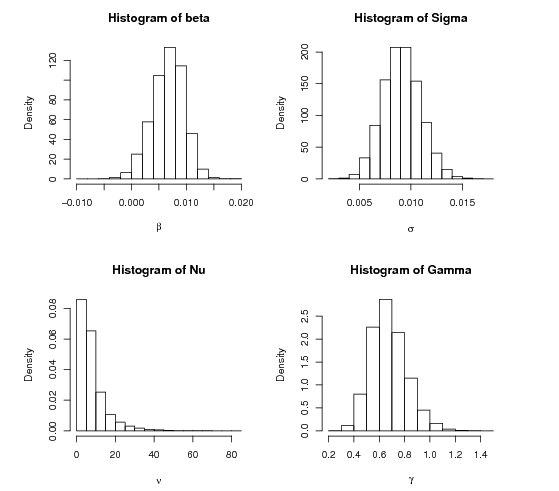}
 \caption{Density histograms of the posteriors in share price return data example.}
 \label{fig:hist_pricedata}

 \centering\includegraphics[width=0.8\textwidth]{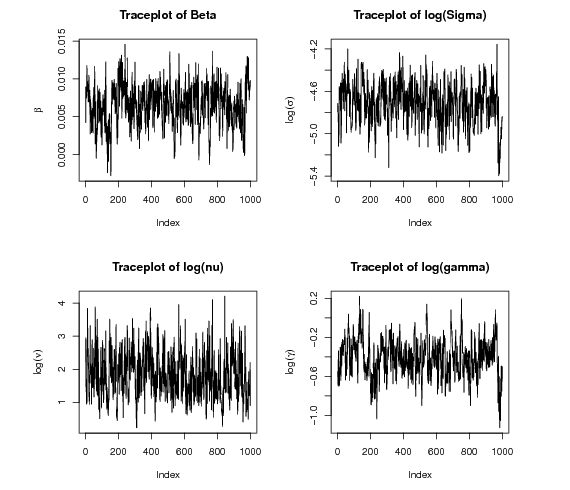}
 \caption{Traceplots of the the last thousand iterations of the RDMH chain for the share price return data example.}
 \label{fig:trace_pricedata}
\end{figure}

\begin{figure}[htp]
 \centering\includegraphics[width=\textwidth]{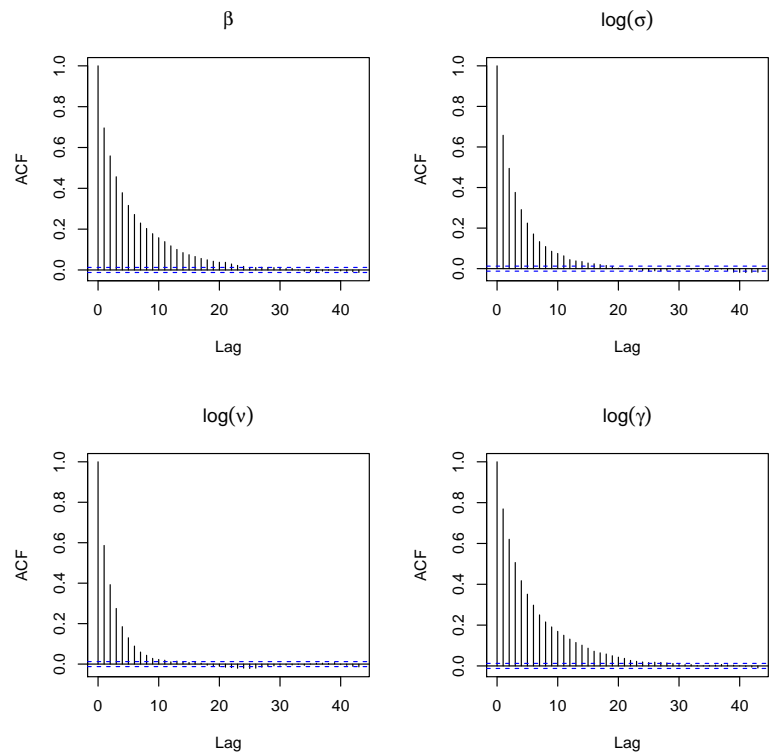}
 \caption{Auto-correlation function of the RDMH sampler.}
 \label{fig:acf_price}
\end{figure}
\begin{table}[htp]
 \begin{center}
  \begin{tabular}{c|cc|cc} \hline\hline
   Parameter & \multicolumn{2}{c|}{RDMH} & \multicolumn{2}{c}{Gibbs} \\ 
             & mean & s.d. & mean & s.d. \\ \hline
  $\b$	& 0.0066 & 0.0029 & $-0.0068$ & 0.0028 \\
  $\s$  & 0.0091 & 0.0018 & 0.0091 & 0.0018 \\ \hline
  \end{tabular}
 \caption{Posterior summaries for $\b$ and $\s$ for the RDMH and Gibbs chains. The result for the Gibbs chain are taken from \cite{fernandez1998}.}
\label{tab:post_beta_sigma}
\vspace{0.5cm}
\begin{tabular}{ccc} \hline \hline
 Parameter & mean(RDMH) & s.d. (RDMH) \\ \hline
 $\nu$ & 8.0119 & 7.0520 \\
 $\g$ & 0.6745 & 0.1408 \\ \hline
\end{tabular}
\caption{Posterior summaries for $\nu$ and $\g$ for the RDMH chain.}
\label{tab:post_nu_gam}
 \end{center}
\end{table}

\cite{fernandez1998} used a Gibbs sampler approach with data-augmentation. They faced some numerical difficulties and perturbed the $y_i$'s slightly to resolve the numerical problems. The RDMH sampler, however, did not face any numerical problem. The results obtained by RDMH differs from the Gibbs sampler perhaps due to this reason. Actually, the posteriors of $\s,~\nu$ and $\g$ were same whether we used Gibbs sampler or not (see the paper by \citep{fernandez1998} for the Gibbs sampler output). The posterior of $\b$ were quite dissimilar for the RDMH and Gibbs samplers. For the Gibbs sampler the posterior was mainly concentrated between $-0.016$ and $0.002$ while it was concentrated between $-0.005$ and $0.015$ for the RDMH chain. Clearly, the Gibbs sampler fails to cover the long tail of the posterior of $\b$ while the RDMH explores it quite easily.

To ensure we also ran a random walk MH sampler and found that the results for the RWMH sampler coincided with that of the RDMH sampler. The summaries of the RDMH sampler is given in Table \ref{tab:post_beta_sigma} and \ref{tab:post_nu_gam} and the histograms and traceplots of the same are given in Figure \ref{fig:hist_pricedata} and \ref{fig:trace_pricedata} respectively. The autocorrelation plots of the RDMH chains are given in Figure \ref{fig:acf_price}. We ran the sampler for 160,000 iterations and discarded the first 10,000 samples as burn-ins. We then thinned the remaining 150,000 samples by 5. Convergence was achieved much earlier though. We also found that the mixing for the RDMH sampler was superior to that of the RWMH sampler.

\section{Further works}\label{sec:further}
We conclude this article with some purview of possible extension to higher dimension. Suppose $\pi$ is a density supported on $\R^k$ and $g$ is density on $\Y^k$. Then the algorithm is given in Algorithm \ref{algo:highd}.
\begin{algo}\label{algo:highd} \topline Random dive MH on $\R^k$ \botline \normalfont \sffamily
\begin{itemize}
 \item Input: Initial value $\mathbf x^{(0)}$ with no component equal to 0, and number of iterations $N$. 
 \item For $t=0,\ldots,N-1$
\begin{enumerate}
 \item Generate $\be \sim g(\cdot)$ and $u_i \sim$ U$(0,1),~i=1,\ldots,k$ independently
 \item For each $i$ if $0 < u_i < 1/2$, set $x'_i = x^{(t)}_i\e_i$. else set $x_i' = x_i^{(t)}/\e_i$
 \item Let $I = \{ 1\leq i \leq k : u_i < 1/2 \}$ and set
\[ \a( \mathbf x ,\be) = \min\left\{ \dfrac{\pi(\mathbf x')}{\pi(\mathbf x)}\dfrac{\prod_{i\in I}\e_i}{\prod_{j\notin I}\e_j},~ 1\right\}\]
\item Set \[\mathbf x^{(t+1)} = \left\{\begin{array}{ccc}
 \mathbf x' & \textsf{ with probability } & \a( \mathbf x^{(t)} ,\be) \\
 \mathbf x^{(t)}& \textsf{ with probability } & 1 - \a( \mathbf x^{(t)} ,\be)
\end{array}\right.\]
\end{enumerate}
\item End for
\end{itemize}
\botline \rmfamily
\end{algo}
This algorithm is still irreducible and aperiodic. It is also Harris recurrent on $\R^{*k}$ and every compact subset of $\R^{*k}$ is still small. The proof is along the same line as Theorem \ref{thm:ergodicity} and \ref{thm:smallset}. Geometric ergodicity is, however, a property that requires a different approach. It is expected that geometric ergodicity of this algorithm still holds for a large class of densities (especially the thick-tailed ones) on higher dimensions. We hope that this article would draw attention of the researchers and the question regarding geometric ergodicity in higher dimension situation would be settled.

The proposal density $g(\be)$ on $\Y^k$ can be chosen to be the product of proposal densities on $\Y$. In such a case, one should choose the univariate proposals which generate $\e$'s close to 1 with high probabilities each (for example, the mixture proposals in Section \ref{sec:appli:pricedata}). This will ensure that the proposed states are not too far away from the current state (in $\R^k$) to reduce the acceptance rate significantly.

\section*{Acknowledgement}
Thanks are due to two anonymous referees whose critical reading of the manuscript and constructive comments lead to major improvement of the paper.


\vspace{1cm}
\linespread{1}
\rmfamily
\emph{Address for correspondence}: \\ Department of Statistics \\ University of Chicago \\ 5734 S. University Avenue \\ Chicago, IL 60637. USA \\ E-mail: sdutta@galton.uchicago.edu
\end{document}